\documentclass[11pt]{article}
\usepackage[margin=1in]{geometry}
\usepackage{amsmath,amssymb} 
\usepackage{graphicx}
\usepackage{wrapfig}
\usepackage[usenames]{color}
\usepackage{url}
\newtheorem{theorem}{Theorem}
\newtheorem{lemma}[theorem]{Lemma}
\newtheorem{definition}[theorem]{Definition}

\newcommand{\qed}{\hfill\rule{1ex}{1em}\penalty-1000{\par\medskip}}

\newcommand{\Substr}{\mathit{Substr}}

\newcommand{\Occ}{\mathit{Occ}}

\newcommand{\SA}{\mathit{SA}}

\newcommand{\lcp}{\mathit{lcp}}
\newcommand{\LCP}{\mathit{LCP}}

\newcommand{\suflen}{\mathit{slen}}
\newcommand{\prelen}{\mathit{plen}}

\newcommand{\derive}{\mathit{val}}

\newcommand{\height}{\mathit{height}}

\newcommand{\llink}{\mathit{Llink}}
\newcommand{\rlink}{\mathit{Rlink}}

\definecolor{bl}{gray}{0.5}
\newcommand{\full}[2]{#1}
\newcommand{\ignore}[1]{}

\newcommand{\logstar}{\log^*\hspace{-.9mm}}

\begin{document}
\title{
  Restructuring Compressed Texts without Explicit Decompression
}
\author{
  Keisuke~Goto$^1$\quad Shirou~Maruyama$^1$\quad Shunsuke~Inenaga$^1$\quad
  Hideo~Bannai$^1$\\ Hiroshi Sakamoto$^2$\quad Masayuki~Takeda$^1$\\
{$^1$ Kyushu University, Japan}\\
{$^2$ Kyushu Institute of Technology, Japan}\\
{\texttt{shiro.maruyama@i.kyushu-u.ac.jp}}\\
{\texttt{\{keisuke.gotou,bannai,inenaga,takeda\}@inf.kyushu-u.ac.jp}}\\
{\texttt{hiroshi@ai.kyutech.ac.jp}}}
\date{}
\maketitle
\begin{abstract}  
  We consider the problem of {\em restructuring} compressed texts
  without explicit decompression.
  We present algorithms which allow conversions from
  compressed representations of a string $T$
  produced by any grammar-based compression algorithm,
  to representations produced by several specific compression
  algorithms including LZ77, LZ78, run length encoding, and
  some grammar based compression algorithms.
  These are the first algorithms that achieve running times polynomial
  in the size of the compressed input and output representations of $T$.
  Since most of the representations we consider can achieve exponential
  compression, our algorithms are theoretically faster in the worst
  case, than any
  algorithm which first decompresses the string for the conversion.
\end{abstract}
\thispagestyle{empty}
\clearpage
\setcounter{page}{1}
\section{Introduction}
Data compression is an indispensable technology for 
the handling of large scale data available today.
The traditional objective of compression has been to save
storage and communication costs,
whereas actually using the data normally requires a
decompression step which can require enormous computational
resources.
However, recent advances in {\em compressed string processing} algorithms
give us an intriguing new perspective in which compression can be regarded as a
form of pre-processing which not only reduces space requirements for
storage, but allows efficient processing of the
strings,
including
compressed pattern matching~\cite{lifshits07:_proces_compr_texts,takanori11:_faster_subseq_dont_care_patter,gawrychowski11:_optim_lzw,gawrychowski11:_patter_lempel_ziv},
string indices~\cite{navarro07:_compr,claudear:_self_index_gramm_based_compr,kreft11:_self_based_lz77},
edit distance and its variants~\cite{Cormode07,hermelin09:_unified_algor_accel_edit_distan,tiskin11:_towar},
and various other applications~\cite{GasieniecSWAT96,inenaga09:_findin_charac_subst_compr_texts,goto10:_fast_minin_slp_compr_strin,matsubara_tcs2009,philip11:_random_acces_gramm_compr_strin}.
These methods assume a compressed representation of the text as input, and process them without
explicit decompression.
An interesting property of these methods is that they can be theoretically 
-- and sometimes even practically -- faster than algorithms which work
on an uncompressed representation of the same data.

The main focus of this paper is to develop a framework in which
various processing on strings can be conducted entirely in the 
world of compressed representations.
A primary tool for this objective is {\em restructuring},
or conversion, of the compressed representation.
Key results for this problem were obtained independently by
Rytter~\cite{rytter03:_applic_lempel_ziv} and 
Charikar {\em et al.}~\cite{charikar05:_small_gramm_probl}:
given a non-self referential LZ77-encoding of size $n$ that
represents a string of length $N$,
they gave algorithms for constructing a balanced grammar
of size at most $O(n\log (N/n))$
in output linear time.
The size of the resulting grammar is an $O(\log (N/g))$ approximation of the
smallest grammar whose size is $g$.
Grammars are generally easier to handle than the LZ-encodings, for example,
in compressed pattern
matching~\cite{gawrychowski11:_patter_lempel_ziv}, 
and this result is the motivational backbone of
many efficient algorithms on grammar compressed strings.

\noindent {\textbf{Our Results:}} In this paper, 
we present a comprehensive collection of new algorithms
for restructuring to and from compressed texts
represented in terms of run length encoding (RLE), LZ77 and LZ78 encodings, 
grammar based compressor RE-PAIR and BISECTION,
edit sensitive parsing (ESP),
straight line programs (SLPs), and admissible grammars.
All algorithms achieve running times polynomial
in the size of the compressed input and output representations of the string.
Since (most of) the representations we consider can achieve exponential
compression, our algorithms are theoretically faster in the worst
case, than any
algorithm which first decompresses the string for the conversion.
Figure~\ref{fig:results_summary} summarizes our results.
Our algorithms immediately allow the following
applications to be solvable in polynomial time in the compressed
world:

\begin{figure}
  \begin{center}
\includegraphics[width=0.9\textwidth]{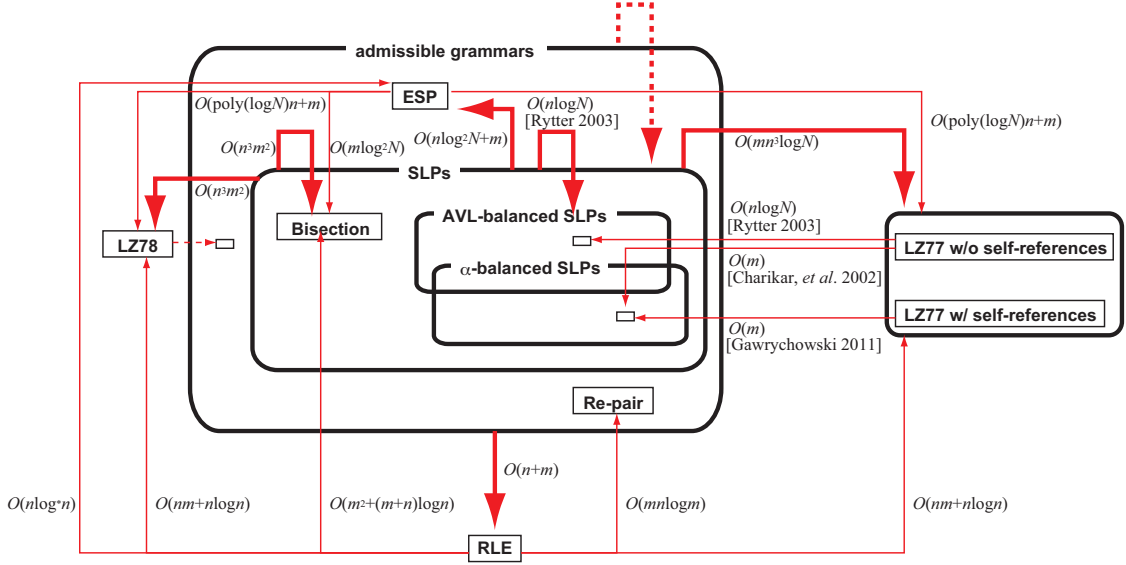}
  \end{center}
  \caption{Summary of transformations between compressed
    representations.
    The label of each arc shows the time complexity of each transformation,
    where $n$ and $m$ are respectively the input and output sizes of each transformation,
    and $N$ is the length of the uncompressed string.
    The broken arcs mean naive $O(n)$-time transformations.
    Complexities without references are results shown in this paper.
  }
  \label{fig:results_summary}  
\end{figure}

\noindent\textbf{Dynamic compressed texts:}
  Although data structures for {\em dynamic} compressed texts
  have been studied somewhat in the
  literature~\cite{chan05:_dynam,gonzalez09:_rank,chan07:_compr,makinen08:_dynam,russo08:_dynam_fully_compr_suffix_trees,lee09:_dynam},
  grammar based or LZ77 compression have not been considered
  in this perspective.
  It has recently been argued that for {\em highly repetitive} strings,
  grammar based compression and LZ77 compression algorithms are better suited
  and achieve better
  compression~\cite{claudear:_self_index_gramm_based_compr,kreft11:_self_based_lz77}.

  Modification of the grammar corresponding to edit
  operations on the string can be conducted in $O(h)$
  time, where $h$ is the height of the grammar.
  (Note that when the grammar is balanced, $h =O(\log N)$ even in the
  worst case.)
  However, these modifications are {\em ad-hoc}, and do not assure that the
  resulting
  grammar is a {\em good} compressed representation of the string,
  and repeated edit operations will inevitably cause degradation on the
  compression ratio.
  Note that previous work of Rytter and Charikar {\em et al.} are not
  sufficient in this respect:
  their algorithms can balance an arbitrary grammar, but they
  must be given an LZ-encoding of the modified string in order for the
  grammar to be small.

\noindent\textbf{Post-selection of compression format:}
  Some methods in the field of data mining and machine learning
  utilize compression as a means of detecting and extracting
  meaningful information from 
  string data~\cite{cormode05:_subst,Cilibrasi05clusteringby}.
  Compression of a given string is achieved by exploiting various
  regularities contained in the string, and
  since different compression algorithms capture different
  regularities, the usefulness of a specific representation will vary
  depending on the application.
  As it is impossible to predetermine the {\em best} compression
  algorithm for all future applications,
  conversion of the representation is an essential task.

  For example, the normalized compression
  distance (NCD)~\cite{Cilibrasi05clusteringby}
  between two strings $X$ and $Y$ 
  with respect to compression algorithm $A$ is
  defined by the values $C_A(XY)$, $C_A(X)$, and $C_A(Y)$ which
  respectively denote the sizes of the compressed representation of 
  strings $XY$, $X$, and $Y$ when compressed by algorithm $A$.
  Restructuring enables us to solve, in the compressed world,
  the problem of calculating the NCD with respect to some compression
  algorithm, given strings which were compressed previously by 
  a (possibly) different compression algorithm.

\section{Preliminaries}

\subsection{Notations}

Let $\Sigma$ be a finite {\em alphabet}.
An element of $\Sigma^*$ is called a {\em string}.
The length of a string $S$ is denoted by $|S|$. 
The empty string $\varepsilon$ is a string of length 0,
namely, $|\varepsilon| = 0$.
For a string $S = XYZ$, $X$, $Y$ and $Z$ are called
a \emph{prefix}, \emph{substring}, and \emph{suffix} of $S$, respectively.
The set of all substrings of a string $S$ is denoted by $\Substr(S)$.
The $i$-th character of a string $S$ is denoted by 
$S[i]$ for $1 \leq i \leq |S|$,
and the substring of a string $S$ that begins at position $i$ and
ends at position $j$ is denoted by $S[i:j]$ for $1 \leq i \leq j \leq |S|$.
For convenience, let $S[i:j] = \varepsilon$ if $j < i$.
For any strings $S$ and $P$,
let $\Occ(S,P)$ be the set of occurrences of $P$ in $S$, i.e.,
$\Occ(S,P) = \{k > 0 \mid S[k:k+|P|-1] = P\}$.

We shall assume that the computer word size is at least $\log |S|$, 
and hence, values representing lengths and positions of $S$
in our algorithms can be manipulated in constant time.

\subsection{Suffix Arrays and LCP Arrays}
The suffix array $\SA$~\cite{manber93:_suffix} of any string $S$
is an array of length $|S|$ such that
$\SA[i] = j$, where $S[j:|S|]$ is the $i$-th lexicographically smallest suffix of $S$.
Let $\lcp(S_1, S_2)$ is the length of the longest common prefix of
$S_1$ and $S_2$.
The \emph{lcp} array of any string $S$ is an array of length $|S|$ such that
$\LCP[i]$ is $\lcp(S[\SA[i-1]:|S|],S[\SA[i]:|S|])$ for $2 \leq i \leq |S|$, 
and $\LCP[1] = 0$.
The suffix array for any string $S$
can be constructed in $O(|S|)$ 
time~(e.g.~\cite{Karkkainen_Sanders_icalp03})
assuming an integer alphabet.
Given the string and suffix array, the lcp array can also be calculated
in $O(|S|)$ time~\cite{Kasai01}.

\subsection{Run Length Encoding}

\begin{definition}
The Run-Length (RL) factorization of a string $S$ is
the factorization $f_1,\ldots,f_n$ of $S$ such that for every $i=1,\ldots,n$,
factor $f_i$ is the longest prefix of $f_i \cdots f_n$ with $f_i\in F$,
where
$F=\bigcup_{a\in\Sigma}\{ a^p \mid p>0 \}$.
\end{definition}
We note that each factor $f_i$ can be written as $f_i=a_i^{p_i}$ for some symbol 
$a_i\in\Sigma$ and some integer $p_i>0$ and
the repeating symbols $a_i$ and $a_{i+1}$
of consecutive factors $f_i$ and $f_{i+1}$ are different.
The output of RLE is a sequence of pairs of symbol $a_i$ and integer $p_i$.
\full{The number of distinct bigrams occurring in $S$ is at most $2n-1$,
since these are $\{ a_i a_i \mid 1 \leq i \leq n\}\cup\{ a_i a_{i+1}
\mid 1 \leq i < n\}$.}

\subsection{LZ Encodings}
LZ encodings are dynamic dictionary based encodings.
There are two main variants for LZ encodings,
LZ78 and LZ77.

The LZ78 encoding~\cite{LZ78} has several variants.
One most popular variant would be the LZW encoding~\cite{LZW}, 
which is based on the LZ78 factorization defined below.
\begin{definition}[LZ78 factorization]
The LZ78-factorization of a string $S$ is
the factorization $f_1,\ldots,f_n$ of $S$ where for every $i=1,\ldots,n$,
factor $f_i$ is the longest prefix of $f_i \cdots f_n$ with $f_i\in F_i$, 
where
$F_i$ is defined by $F_1=\Sigma$ and $F_{i+1} = F_i\cup\{f_i f_{i+1}[1]\}$.
\end{definition}
The output is the sequence of IDs of factors $f_i$ in $F_i$.
We note that $F_i$ can be recovered from this sequence and thus is
not included in the output.

The LZ77 encoding~\cite{LZ77} also has many variants.
The LZSS encoding~\cite{LZSS} is based on 
the LZ77 factorization below.
The LZ77 factorization has two variations depending upon
whether self-references are allowed.
\begin{definition}[LZ77 factorization w/o self-references]
The LZ77-factorization without self-references of a string $S$ is
the factorization $f_1,\ldots,f_n$ of $S$ such that for every $i=1,\ldots,n$,
factor $f_i$ is the longest prefix of $f_i \cdots f_n$ with $f_i\in F_i$,
where
$F_i=\Substr(f_1\cdots f_{i-1})\cup\Sigma$.
\end{definition}

\begin{definition}[LZ77 factorization w/ self-references]
The LZ77-factorization with self-references of a string $S$ is
the factorization $f_1,\ldots,f_n$ of $S$ such that for every $i=1,\ldots,n$,
factor $f_i$ is the longest prefix of $f_i \cdots f_n$ with $f_i\in F_i$,
where
$F_i=\Substr(f_1\cdots f_{i-1}f_i^{\prime})\cup\Sigma$, where
$f_i^{\prime}$ is the prefix of $f_i$ obtained by removing the last symbol.
\end{definition}
The LZSS is based on the LZ77 with self-references and
its output is a sequence of pointers to factors $f_i$.

\subsection{Grammar-based compression methods}
An \emph{admissible grammar}~\cite{Kieffer00admissible} 
is a context-free grammar that generates a single string. 

\subsubsection{Re-pair}
Starting with $w_1=S$,
we repeat the following until no bigrams occur more than once in $w_i$:
we find a most frequent bigram $\gamma_i$ in the string $w_i$,
and then replace every non-overlapping occurrence of $\gamma_i$ in $w_i$
with a new variable $X_i$ to obtain string $w_{i+1}$.
Let $r$ be the number of iterations.
The resulting grammar has the production rules of
$\{ X_i \to \gamma_i \}_{i=1}^{r} \cup \{ X_{r+1} \to w_{r+1} \}$.

\begin{theorem}[\cite{charikar05:_small_gramm_probl}]
For any string $S$ of length $N$,
Re-pair constructs 
in $O(N)$ time
an admissible grammar of size $O(g(N/ \log N)^{2/3})$,
where $g$ is the size of the smallest grammar that derives $S$.
\end{theorem}

\subsubsection{Bisection}
\label{subsubsection:bisection}
The Bisection algorithm~\cite{Kieffer00admissible,Kieffer00MPM}
constructs a grammar that can be described recursively as follows:
the variable representing string $S$ ($|S|\geq 2$)
is derived by the rule $X \rightarrow YZ$, with
$|Y| = 2^k$ and $|Z| = |X| - 2^k$, where $k$ is the largest integer
s.t. $2^k < |X|$.
The production rules for $S[1:2^k]$ and $S[2^k+1:|S|]$ are defined recursively.
Whenever
$S[i:i+q-1] = S[j:j+q-1]$ for some $i,j,q \geq 1$ which appear in the
above construction, the same variable is to be used for deriving these substrings.

\begin{theorem}[\cite{charikar05:_small_gramm_probl}]
For any string $S$ of length $N$,
Bisection constructs an admissible grammar of size $O(g(N/ \log N)^{1/2})$,
where $g$ is the size of the smallest grammar that derives $S$.
\end{theorem}

\subsection{Edit-sensitive parsing (ESP)}
A string $a^k$ ($k\geq 2$) is called 
a repetition of symbol $a$, and $a^+$ is its abbreviation.
We let $\log^{(1)}n=\log n$, $\log^{(i+1)}=\log\log^{(i)}n$,
and $\logstar n=\min\{i\mid \log^{(i)}n\leq 1\}$.
For example, $\logstar n\leq 5$ for any $n\leq 2^{65536}$.
We thus treat $\logstar n$ as a constant for sufficiently large $n$.

We assume that any context-free grammar $G$ is 
{\em admissible}, i.e., $G$ derives just one string and
for each variable $X$, exactly one production rule $X\to \alpha$ exists.
The set of variables is denoted by $V(G)$, and
the set of production rules, called dictionary, is denoted by $D(G)$.
We also assume that $X\to\alpha\in D(G)$ and $Y\to\alpha\in D(G)$ 
implies $X=Y$ because one of them is unnecessary.
We use $V$ and $D$ instead of $V(G)$ and $D(G)$ when $G$ is omissible.
The string derived by $D$ from a string $S\in(\Sigma\cup V)^*$
is denoted by $S(D)$.
For example, when $S=aYY$ and $D=\{X\to bc,Y\to Xa\}$,
we obtain $S(D)=abcabca$.

For any string, it is uniquely partitioned to
$w_1a^+_1w_2a^+_2\cdots w_ka^+_kw_{k+1}$ by maximal repetitions, 
where each $a_i$ is a symbol and $w_i$ is a string containing no repetition.
Each $a^+_i$ is called Type1 metablock,
$w_i$ is called Type2 metablock if $|w_i|\geq \log^*n$,
and other short $w_i$ is called Type3 metablock,
where if $|w_i|=1$, this is attached to $a^+_{i-1}$ or $a^+_i$,
with preference $a^+_{i-1}$ when both are possible.
Thus, any metablock is longer than or equal to two.

Let $S$ be a metablock and $D$ be a current dictionary
starting with $D=\emptyset$.
We set $ESP(S,D)=(S',D\cup D')$ for 
$S'(D')=S$ and $S'$ described as follows:

\begin{enumerate}
\item When $S$ is Type1 or Type3 of length $k\geq 2$,
\begin{enumerate}
\item If $k$ is even, let $S'=t_1t_2\cdots t_{k/2}$, and
make $t_i\to S[2i-1:2i]\in D'$.
\item If $k$ is odd, 
let $S'=t_1t_2\cdots t_{(k-3)/2}\: t$, and make
$t_i\to S[2i-1:2i]\in D'$ and $t\to S[k-2:k]\in D'$
where $t_0$ denotes the empty string for $k=3$.
\end{enumerate}
\item When $S$ is Type2,
\begin{enumerate}
\item[(c)] 
for the partitioned $S=s_1s_2\cdots s_k$ $(2\leq |s_i|\leq 3)$
by {\em alphabet reduction}, 
let $S'=t_1t_2\cdots t_k$, and make $t_i\to s_i\in D'$. 
\end{enumerate}
\end{enumerate}

Cases (a) and (b) denote a typical {\em left aligned parsing}.
For example, in case $S=a^6$, $S'=x^3$ and $x\to a^2\in D'$,
and in case $S=a^9$, $S'=x^3y$ and $x\to a^2,y\to aaa\in D'$.
In Case (c), we omit the description of 
alphabet reduction~\cite{Cormode07} 
because the details are unnecessary in this paper.

\full{
Case (b) is illustrated in Fig.~\ref{type1} for a Type1 string,
and the parsing manner in Case (a) is obtained by
ignoring the last three symbols in Case (b).
Parsing for Type2 is analogous.
Case (c) is illustrated in Fig.~\ref{type2}.
}{}

\full{
\begin{figure}[bt]
\begin{center}
\includegraphics[scale=.5]{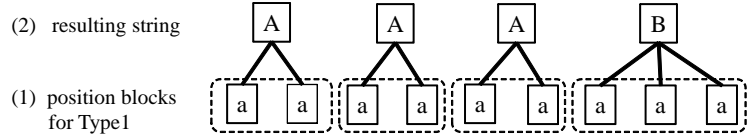}
\end{center}
\vspace{-.5cm}
\caption{
Parsing for Type1 string:
Line (1) is an original Type1 string $S=\mbox{a}^9$
with its position blocks.
Line (2) is the resulting string $\mbox{AAAB}$, and
the production rules $\mbox{A}\to \mbox{aa}$ and $\mbox{B}\to \mbox{aaa}$.
Any Type3 string is parsed analogously.
}
\label{type1}
\end{figure}
}
\full{
\begin{figure}[tb]
\begin{center}
\includegraphics[scale=.5]{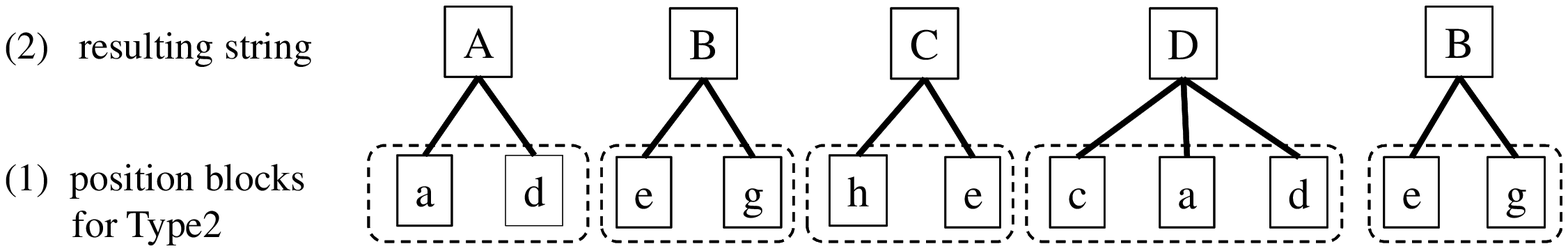}
\end{center}
\vspace{-.5cm}
\caption{
Parsing for Type2 string:
Line (1) is an original Type2 string `$\mbox{adeghecadeg}$'
with its position blocks by alphabet reduction
where its definition is omitted in this paper.
Line (2) is the resulting string $\mbox{ABCDB}$, and
the production rules $\mbox{A}\to \mbox{ad}$, $\mbox{B}\to \mbox{eg}$, etc.
}
\label{type2}
\end{figure}
}

Finally, we define ESP for any string $S\in (\Sigma\cup V)^*$ 
that is partitioned to $S_1S_2\cdots S_k$ by $k$ metablocks;
$ESP(S,D) 
= (S',D\cup D')
= (S'_1\cdots S'_k,D\cup D')$,
where $D'$ and each $S'_i$ satisfying $S'_i(D')=S_i$ are defined in the above. 

Iteration of ESP is defined by $ESP^i(S,D) = ESP^{i-1}(ESP(S,D))$.
In particular, $ESP^*(S,D)$ denotes 
the iterations of ESP until $|S|=1$.
After computing $ESP^*(S,D)$,
the final dictionary represents a rooted ordered binary tree
deriving $S$, which is denoted by $ET(S)$.

\begin{lemma}\label{Cormode1}\rm(Cormode and Muthukrishnan~\cite{Cormode07})
The height of $ET(S)$ is $O(\log |S|)$ and 
$ET(S)$ can be computed in time $O(|S|\logstar |S|)$ time.
\end{lemma}

\begin{lemma}\label{Cormode2}\rm(Cormode and Muthukrishnan~\cite{Cormode07})
Let $S=s_1s_2\cdots s_k$ be the partition of a Type2 metablock $S$
by alphabet reduction.
For any $1\leq j\leq |S|$,
the block $s_i$ containing $S[j]$ is determined by
at most $S[j-\logstar N-5:j+5]$.
\end{lemma}

We refer to another characteristic of ESP for
pattern embedding problem.
Nodes $v_1,v_2$ in $T=ET(S)$ are {\em adjacent} in this order 
if the subtrees on $v_1,v_2$ are adjacent in this order.
A string $p_1\cdots p_k$ of length $k$ is embedded in $T$
if there exist nodes $v_1,\ldots, v_k$ such that 
${\it label}(v_i)=p_i$ and any $v_i,v_{i+1}$ are adjacent in this order.
If $T[i]$, the $i$-th leaf of $T$, is the leftmost leaf of $v_1$
and $T[j]$ is the rightmost leaf of $v_k$,
we call that $p_1\cdots p_k$ is embedded as $T[i:j]$.

\begin{definition}\rm
$Q\in(\Sigma\cup V)^*$ is called an {\em evidence} of $P\in\Sigma^*$ in $S$
if the following holds:
$S[i:j]=P$ iff $Q$ is embedded as $T[i:j]$.
\end{definition}

We note that any $P$ has at least one evidence since
$P$ itself is an evidence of $P$.

\begin{lemma}\label{spire2011}\rm (Maruyama et al.~\cite{Maruyama11})
Given $T=ET(S)$, 
for any $T[i:i+t]=P$, there exists an evidence 
$Q=q_1\cdots q_k$ of $P$ with maximal repetitions $q_\ell$ and $k=O(\log t)$.
We can compute the $Q$ in $O(\log t\log|S|)$ time, 
and we can also check if $Q$ is embedded as 
$T[j:j+t]$ in $O(\log t\log |S|)$ time for any $j$.
\end{lemma}

\subsection{Straight Line Programs}

\full{
\begin{figure}
\centerline{\includegraphics[width=0.5\textwidth]{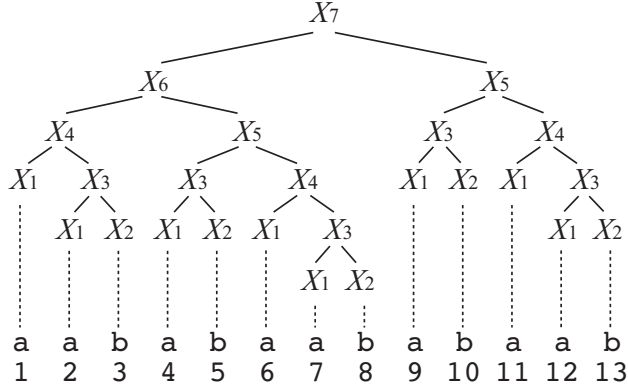}}
\caption{
  The derivation tree of
  SLP with $X_1 \rightarrow
  \mathtt{a}$, 
  $X_2 \rightarrow \mathtt{b}$, $X_3 \rightarrow X_1X_2$,
  $X_4 \rightarrow X_1X_3$, $X_5 \rightarrow X_3X_4$, $X_6 \rightarrow X_4X_5$, and $X_7 \rightarrow X_6X_5$,
  representing string $S = \derive(X_7) = \mathtt{aababaababaab}$.
}
\label{fig:SLP}
\end{figure}
}{}

A {\em straight line program} ({\em SLP})~\cite{NJC97} is a widely accepted 
abstract model of outputs of grammar-based compressed methods.
An SLP is a sequence of assignments 
$\{X_i \rightarrow expr_i\}_{i=1}^n$,
where each $X_i$ is a variable and each $expr_i$ is an expression, where
$expr_i = a$ ($a\in\Sigma$), or $expr_i = X_{\ell} X_r$~($\ell,r < i $).
Namely, SLPs are admissible grammars in the Chomsky normal form,
and hence outputs of admissible grammars can be easily converted to 
SLPs in linear time (see also Figure~\ref{fig:results_summary}).
Let $\derive(X_i)$ represent the string derived from $X_i$.
When it is not confusing, we identify a variable $X_i$
with $\derive(X_i)$.
Then, $|X_i|$ denotes the length of the string $X_i$ derives.
An SLP $\{X_i\rightarrow expr_i\}_{i=1}^n$ {\em represents} 
the string $S = \derive(X_n)$.
The \emph{size} of an SLP is the number of
assignments in it.
The {\em height} of variable $X_i$ is denoted $\height(X_i)$,
and is 1 if $X_i = a~(a\in \Sigma)$,
and $1 + \max\{ \height(X_\ell), \height(X_r)\}$ if $X_i = X_\ell X_r$.
The height of an SLP $\{X_i \rightarrow expr_i\}_{i=1}^n$ 
is defined to be $\height(X_n)$.

Note that $|X_i|$ and $\height(X_i)$ for all variables can be calculated in
a total of $O(n)$ time by simple dynamic programming iterations.
In the rest of the paper, we will therefore assume
that these values will be available.

The following results are known for SLP compressed strings:
\begin{theorem}[\cite{lifshits07:_proces_compr_texts}]
  \label{theorem:fcpm}
  Given two SLPs of total size $n$ that describe strings $S$ and $P$, respectively,
  a succinct representation of $\Occ(S, P)$ can be computed in $O(n^3)$ time
  and $O(n^2)$ space.
\end{theorem}
Since we can compute $|S|$ and $|P|$ in $O(n)$ time
and a membership query to the succinct representation can be answered in $O(n)$ time~\cite{MasamichiCPM97},
the equality checking of whether $S = P$ can be done in a total of $O(n^3)$ time.

\begin{lemma}[\cite{lifshits:DSP:2006:798}]
  \label{lemma:slpsubstring}
  Given an SLP of size $n$ representing string $S$,
  an SLP of size $O(n)$ which represents an arbitrary substring $S[i:j]$ 
  can be constructed in $O(n)$ time.
\end{lemma}

\section{Algorithms for Restructuring Compressed Texts}

In this section we present our polynomial-time algorithms 
that converts an input compressed representation to another compressed representation.
In the sequel, $n$ and $m$ will denote the sizes of 
the input and output compressed representations, respectively.

\subsection{Conversions from Run Length Encoding}
For conversions from Run Length Encodings,
we obtain the results below.
\begin{theorem}[Run Length Encoding to Re-pair]
  Given an RL factorization of size $n$ that represents string $S$,
  the grammar of size $m$ produced by applying Re-pair algorithm to $S$
  can be computed in $O(nm \log m)$ time.
\end{theorem}
\begin{proof}
\full{
We consider a simple simulation of the Re-pair algorithm
that works on the RL factorization of the string $S$.
We shall assume that the Re-pair algorithm replaces non-overlapping bigrams 
with a new variable in a left-first manner.
Let $Y_i \rightarrow Y_{\ell}Y_{r}$ denote the $i$-th rule
produced by the Re-pair algorithm running on $S$.
Let $S_1 = S$, and for $i \geq 1$ let $S_{i+1}$ denote the string obtained by 
replacing frequent bigrams by $Y_1$, $Y_2$, \ldots, and $Y_{i}$.
Note that the bigram $Y_{\ell}Y_{r}$ will not occur in $S_{i+1}$.
Consider the RL factorization of $S_{i}$,
and let $w_i$ denote the string obtained by concatenating 
the RL factors of $S_i$ 
consisting of characters in $\Sigma \cup \{Y_j\}_{j=1}^{i-1}$.

We find the most frequent bigram $Y_\ell Y_r$ in $w_i$,
and then replace non-overlapping occurrence of $Y_\ell Y_r$ in $w_i$
with a new variable $Y_i$ on the left priority basis,
and then compute $w_{i+1}$.

Let $a,b,c \in \Sigma$ with $a \neq b$ and $a \neq c$,
and let $ba^pc$ be a substring of the original string $S$,
where $p \geq 1$.
Consider any occurrence of $ba^pc$ that begins at position $v$ in $S$,
namely, let $S[v:v+p+1] = ba^pc$.
There are two cases to consider:
(1) the range $[v:v+p+1]$ is fully contained within a variable $Y_k$ in $w_i$;
(2) the range $[v:v+p+1]$ is contained in a substring of $w_i$ of form
$(Y_{s})^{e} (Y_{k(1)})^{q} Y_{k(2)} \cdots Y_{k(l)} (Y_{r})^{t}$
with $k(1) > k(2)> \cdots > k(l)$,
where $\derive(Y_{k(1)})^{q} \cdot \derive(Y_{k(2)}) \cdots \derive(Y_{k(l)}) = a^{p^\prime}$
for some $p^\prime \leq p$,
$ba^x$ is a suffix of $\derive(Y_{s})$,
$a^yc$ is a prefix of $\derive(Y_{r})$,
and $x + p^\prime + y = p$.

Let $Y_\ell Y_r$ be the most frequent bigram in $w_i$.
It is possible to replace 
non-overlapping occurrences of $Y_\ell Y_r$ in $O(n)$ time, as follows:
We can see that $\derive(Y_\ell)\derive(Y_r)$ occurs 
either (A) in a sequence $f_j f_{j+1} \cdots f_{j+d-1}$ 
of $d \geq 2$ consecutive factors in $w_1$
or (B) entirely within a single factor $f_{j}$ of $w_1$.
This is because, if $\derive(Y_\ell)\derive(Y_r)$ contains 
at least two distinct characters $a \neq b$, then 
it occurs in a sequence of $d$ factors,
and if $\derive(Y_\ell)\derive(Y_r) = a^z$,
then it is fully contained in a factor.
Consider case (A):
Let $\derive(Y_\ell)[1] = c$.
Since the number of factors of form $f = c^p$ does not exceed $n$,
the number of occurrences of the bigram of case (A) is $O(n)$.
Now consider case (B):
According to the observation (2) above,
any bigram $Y_\ell Y_r$ with $\derive(Y_\ell)\derive(Y_r) = a^z$ and $\ell \neq r$
occurs at most once in each substring of $w_i$ that corresponds to a factor $f_j$.
Hence the number of occurrences of such a bigram in $w_i$ is at most $n$.
If $\ell = r$,
then the bigram $Y_\ell Y_\ell$ can occur $q-1$ times at each factor.
We then replace $(Y_\ell)^q$ with $(Y_i)^{q/2}$ if $q$ is even,
and with $(Y_i)^{(q-1)/2}Y_\ell$ otherwise, in $O(1)$ time.

Since each $w_i$ consists of characters in $\Sigma \cup \{Y_j\}_{j=1}^{i-1}$,
the number of all bigrams in $w_i$ is $O(m^2 + m|\Sigma|) = O(m^2)$.
We find the most frequent bigram in $O(\log m)$ time using a heap,
and the total time complexity for converting 
the RL factorization to the grammar corresponding to Re-pair
is $O(nm \log m)$.
\qed}
{Proof in full version (in Appendix)}
\end{proof}

\begin{theorem}[Run Length Encoding to LZ77/LZ88]
  Given an RL factorization of size $n$ that represents string $S$,
  the LZ factorization of $S$ can be computed in $O(nm + n\log n)$ time,
  where $m$ is the size of LZ factorization.
\end{theorem}
\begin{proof}
\full{
Let $a_1^{p_1},\ldots,a_n^{p_n}$ be the RL factorization of a text $S$.
Assume that we have already computed the first $i-1$ 
LZ77 factors, $f_1,\ldots,f_{i-1}$, of $S$.
Let the pair of integers $(u,q)$ satisfy
$q + p_1 + \cdots + p_{u-1} = |f_1\cdots f_{i-1}|$,
where $1 \leq u \leq n$ and $1\leq q\leq p_u$.
For a new factor $f_i$,
compute the lengths $l_j$ of the longest common prefix 
of 
$a_{u+1}^{p_{u+1}}\cdots a_n^{p_n}$ and
each suffix 
$a_j^{p_j}\cdots a_n^{p_n}~(1\leq j\leq u)$
of the RLE, where each RL factor $a^p$
is regarded as single symbol.
The length of the $i$-th LZ77 factor is then:
$\max_j\{ p_{u+1} + \cdots + p_{u+l_j} + P + Q \}$,
where
$P = 0$ if $a_u \neq a_j$ and $P = \min\{p_u -q,p_{j-1}\}$ otherwise,
and
$Q = 0$ if $a_{u+l_j+1} \neq a_{j+l_j}$ and $Q = \min\{p_{u+l_j+1},p_{j+l_j}\}$ otherwise.
The process is then repeated to obtain $f_{i+1}$
from the pair of integers $(u+l_j+1, p_{u+l_j+1} - Q)$.
A na\"ive algorithm for obtaining each $l_j$ costs $O(n)$ time,
and therefore results in an $O(n^2m)$ time algorithm to check each of
the $O(n)$ suffixes to construct the $O(m)$ factors.
If we construct a suffix and lcp array on the RLE string beforehand,
$l_j$ can be computed in
$O(1)$ time, since it amounts to a range minimum query on the lcp array.
Note that the sum $p_{u+1} + \cdots + p_{u+l_j}$ can also be obtained in
constant time with $O(n)$ preprocessing, by constructing
an array $\mathit{sum}[i] = p_1 + \cdots + p_i$ and computing
$\mathit{sum}[u+l_j] - \mathit{sum}[u]$.
Therefore, conversion can be done in $O(nm)$ time provided that the
suffix array and lcp arrays are constructed.
The construction of the arrays require $O(n\log n)$ time,
to sort and number each of character of the alphabet $a_i^{p_i}$.

LZ78 factorization can be achieved by a simple modification.
\qed}
{Proof in full version (in Appendix)}
\end{proof}

\begin{theorem}[Run Length Encoding to Bisection]
  Given an RL factorization of size $n$ that represents string $S$,
  the grammar of size $m$ produced by applying Bisection algorithm to $S$
  can be computed in $O(m^2 + (m+n)\log n)$ time.
\end{theorem}
\begin{proof}
\full{
Consider the following
top-down algorithm which closely follows the description of Bisection in
Section~\ref{subsubsection:bisection}.
Assume we want to construct the children $Y_\ell,Y_r$ of variable
$Y_s$ representing $S[i:j]$, to produce the grammar rule $Y_s
\rightarrow Y_\ell Y_r$.
Note that an arbitrary substring of $S[i:j]$ 
which is contained in the RLE $a_{k-1}^{p_{k-1}}\cdots a_{k+l}^{p_{k+l}}$
can be represented as a 
4-tuple $(x,k,l,y)$, where
$i = p_1 + \cdots + p_{k-1} - x + 1$,
$j = p_1 + \cdots + p_{k+l-1} + y - 1$,
$0 \leq x < p_{k-1}$, and $0\leq y < p_{k+l}$.
Let $k$ represent the largest integer where $2^k < j-i+1$.
For the substring $S[i:j]$ under consideration,
the 4-tuple for substrings $S[i:i+2^k-1]$ and $S[i+2^k:j]$ 
can be obtained in $O(\log n)$ time.
Note that equality checks between substrings represented as 4-tuples
can be conducted in $O(1)$ time with $O(n\log n)$ preprocessing,
using range minimum queries on the lcp arrays,
similar to the technique used in the conversion to LZ encodings.
Equality checks are conducted against the
$O(m)$ variables that will be contained in the output.
If there exist variables which derive
the same string, the existing variables are used in place of $Y_\ell$
and/or $Y_r$, and $Y_\ell$ and/or $Y_r$ will not be contained in the
output.
Since equality checks are conducted only for the children of variables
which are contained in the output, they are conducted only $O(m)$ times.
Therefore, conversion can be done in
$O(n\log n + m (m + \log n)) = O(m^2 + (m+n)\log n)$ time.
\qed}
{Proof in full version (in Appendix)}
\end{proof}

\subsection{Conversions from arbitrary SLP}

\begin{theorem}[SLP to Run Length Encoding]
  Given an SLP of size $n$ that represents string $S$,
  the RL factorization of $S$ can be computed in $O(n+m)$ time and $O(n)$ space,
  where $m$ is the size of the RL factorization.
\end{theorem}
\begin{proof}
For each variable $X_i$, we first compute 
the maximal length of the run of identical characters 
which is a prefix (resp. suffix) of $X_i$,
denoted by $\prelen(X_i)$ (resp. $\suflen(X_i)$).
This can be computed in $O(n)$ time by a simple dynamic programming:
for $X_i \rightarrow X_\ell X_r$, we have $\prelen(X_i) =
\prelen(X_\ell)$ if $\prelen(X_\ell) < |X_\ell|$ or
$X_\ell[|X_\ell|]\neq X_r[1]$,
and $\prelen(X_i) = \prelen(X_\ell) + \prelen(X_r)$ otherwise.
$\suflen(X_i)$ can be computed likewise.

Next, for each variable $X_i\rightarrow X_{\ell} X_{r}$,
let $\llink(X_i)$ denote the variable $X_{i'}\rightarrow X_{\ell'} X_{r'}$
such that $X_{i'}$ is the shallowest descendant of $X_i$ lying on the
left most path of the derivation tree of $X_i$, satisfying
$\suflen(X_{i'}) \leq |X_{r'}|$.
$\llink(X_i)$ can also be computed for all $X_i$ in $O(n)$ time,
by a simple dynamic programming.
$\rlink(X_i)$ can be defined and computed likewise.

The conversion algorithm is then a top down post-order traversal on
the derivation tree of SLP but with jumps using $\llink$ and $\rlink$.
For the root $X_n$,
we output
(1) $X_n[1]^{\prelen(X_n)}$,
(2) the RLE of $X_n$ except for the first and last RL factors of $X_n$,
and
(3) $X_n[|X_n|]^{\suflen(X_n)}$.
(2) can be computed recursively as follows:
at each variable $X_i \rightarrow X_\ell X_r$,
we output 
(2.1) the RLE of $\llink(X_i)$ except for the first and last 
RL factors of $\llink(X_i)$,
(2.2) either 
$X_\ell[|X_\ell|]^{\suflen(X_\ell)}X_r[1]^{\prelen(X_r)}$
if $X_\ell[|X_\ell|] \neq X_r[1]$,
or
$X_r[1]^{\suflen(X_\ell)+\prelen(X_r)}$ if $X_\ell[|X_\ell|] = X_r[1]$,
and
(2.3) the RLE of $\rlink(X_i)$ except for the first and last
RL factors of $\rlink(X_i)$.
The theorem follows since the output of each RL factor is done in constant time.
\qed
\end{proof}

\begin{theorem}[SLP to LZ77]
  Given an SLP of size $n$ that represents string $S$,
  the LZ77 factorization of size $m$
  can be computed in $O(mn^3\log N)$ time.
\end{theorem}
\begin{proof}
Assume we have already computed $f_1, \ldots, f_{i-1}$ of $S$ from
a given SLP of size $n$.
Firstly we consider the LZ77 factorization without self-references.
For a new factor $f_i$, do a binary search on the length of the factor:
create a new SLP of that length, and conduct pattern matching on
the input SLP.
If a match exists in 
the range that corresponds to the previous factors $f_1, \ldots, f_{i-1}$, 
i.e., in the prefix $S[1:\sum_{j}^{i-1}|f_j|]$ of $S$,
then the length of $f_i$ can be longer, and if not, it must be shorter.
Using Theorem~\ref{theorem:fcpm} and Lemma~\ref{lemma:slpsubstring}
the LZ77 factorization of size $m$ can thus be computed in $O(mn^3\log N)$ time.
To compute the LZ77 factorization with self-references,
we search for the longest match that begins at a position from $1$ to $\sum_{j}^{i-1}|f_j|$ in $S$.
The time complexity is the same as above.~\qed
\end{proof}
\begin{theorem}[SLP to LZ78]
  Given an SLP of size $n$ that represents string $S$,
  the LZ78 factorization of size $m$
  can be computed in $O(n^3m^2)$ time.
\end{theorem}
\begin{proof}
Our algorithm for converting an SLP to LZ78 follows a similar idea:
When computing a new factor $f_i$,
we construct a new SLP of $f_{k}f_{k+1}[1]$ for each $1 \leq k < i$,
and run the pattern matching algorithm on the input SLP.
The longest match in the suffix $S[\sum_{j=1}^{i-1}|f_{j}|+1:|S|]$
provides the new factor $f_i$.
By Theorem~\ref{theorem:fcpm} and Lemma~\ref{lemma:slpsubstring},
pattern matching tasks for computing each factor $f_i$
takes $O(n^3m)$ time,
and therefore the total time complexity is $O(n^3m^2)$.~\qed
\end{proof}

\subsection{SLP to Bisection}
\label{subsec:slp2bisection}
\begin{theorem}
  Given an SLP of size $n$ that represents string $S$,
  the grammar of size $m$ produced by applying Bisection algorithm to $S$
  can be computed in $O(n^3m^2)$ time.
\end{theorem}
\begin{proof}
Given an arbitrary SLP of size $n$ representing $S$, consider the following
top-down algorithm which closely follows the description of Bisection in
Section~\ref{subsubsection:bisection}.
Assume we want to construct the children $Y_\ell,Y_r$ of variable
$Y_s$ representing $S[i:j]$, to produce the grammar rule $Y_s
\rightarrow Y_\ell Y_r$.
Let $k$ represent the largest integer where $2^k < j-i+1$.
By using Lemma~\ref{lemma:slpsubstring}, SLPs 
$Y_\ell$ representing $S[i:i+2^k-1]$
and
$Y_r$ representing $S[i+2^k:j]$, can be constructed in $O(n)$ time.
For these SLPs, equality checks are conducted against all $O(m)$ variables
corresponding to variables that will be contained in the output
produced so far. If there exist variables which derive
the same string, the existing variables are used in place of $Y_\ell$
and/or $Y_r$, and $Y_\ell$ and/or $Y_r$ will not be contained in the
output.
From Theorem~\ref{theorem:fcpm}, the equality checks for 
$Y_\ell$ and $Y_r$ can be conducted in a total of $O(n^3m)$ time.
Since equality checks are conducted only for the children of variables
which are contained in the output, the total time is $O(n^3m^2)$.
\qed
\end{proof}
\subsection{Conversions to and from ESP}

Given a representation of SLP $G$ for a string $S$,
we design algorithms to compute LZ77 and LZ78 factorizations
for $S$ without explicit decompression of $G$ 
in $O((n+m)\log^d N)$ time.
Here $n/m$ is the size of input/output grammar size, $N=|S|$,
and $d$ is a constant.
Our method is based on the transformation of any SLP to
its canonical form by way of an equivalent ESP.

\begin{lemma}\label{sorting}
Given a dictionary $D$ from $ESP^*(S,D)$ for some $S\in\Sigma^*$,
and the set $V$ of variables in $D$,
we can compute an SLP with the dictionary $D'$ and the set $V'$ of variables
which satisfies the following conditions:
(1) $|D'|\leq 2|D|$ and
(2) for any $X_i,X_j\in V'$, $\derive(X_i)\leq_{lex} \derive(X_j)$ iff
$i\leq j$, where
$\leq_{lex}$ denotes the lexical order over $\Sigma$.
The computation time is $O(n\log n\log^3 N)$ for
$|V|=n$ and $|S|=N$.
\end{lemma}

\begin{proof}
Consider $T_X=ET(\derive(X))$ and $T_Y=ET(\derive(Y))$ for any $X,Y\in V$.
Let $t=\lfloor |\derive(Y)|/2\rfloor$.
By Lemma~\ref{spire2011},
we can compute an evidence $Q$ of the pattern $T_Y[1:t]$ 
in $O(\log^2t)=O(\log^2 N)$ time.
We can also check if $Q$ is embedded as $T_X[1:t]$ 
in $O(\log^2 N)$ time.
By this binary search, we can find the length of
longest common prefix of $\derive(X)$ and $\derive(Y)$ in $O(\log^2 N)$ time.
Thus, we can sort all variables in $V$
in $O(n\log n\log^3 N)$ time.
After sorting all variables in $V$,
we rename any variable according to its rank.
If there is a variable $X$ with $X\to X_iX_jX_k$,
we divide it to $X\to YX_k$ and $Y\to X_iX_j$
by an intermediate variable $Y$ and
we can determine the rank of such new variables
in additional $O(n\log n)$ time.~\qed
\end{proof}

Dictionaries $D_1,D_2$ of two admissible grammars 
are called {\em consistent} if $X\to\alpha,Y\to\alpha\in D_1\cup D_2$
implies $X=Y$, and 
consistent dictionaries $D_1,\ldots,D_k$ are similarly defined.

For $\alpha\in(\Sigma\cup V)^*$,
$\alpha=q_1\cdots q_k$ is called a 
{\em run-length representation of} $\alpha$ if
each $q_i$ is a maximal repetition of $p_i\in\Sigma\cup V$.
For example, the run-length representation of $abbaaacaa$
is $q_1q_2q_3q_4q_5=ab^2a^3ca^2$.
The number $k$ of $\alpha=q_1\cdots q_k$ is called the change of $\alpha$.

Let $S=\alpha\beta\gamma$ and $S'=\alpha'\beta'\gamma'$ satisfying
$ESP(S,D)=(S',D\cup D')$ with 
$\alpha'(D')=\alpha$, $\beta'(D')=\beta$, and $\gamma'(D')=\gamma$.
Then we call such $S=\alpha\beta\gamma$ 
a {\em stable decomposition} of $S$.
An expression $ESP(\alpha[\beta]\gamma,D)=(\alpha'[\beta']\gamma',D\cup D')$
denotes an ESP to replace the $\alpha/\beta/\gamma$ to
the $\alpha'/\beta'/\gamma'$, respectively.
For a string $\alpha$, 
$\overline{\alpha}$ and $\underline{\alpha}$ denote a prefix of
$\alpha$ and a suffix of $\alpha$, respectively.

\begin{lemma}\label{mergingESP}
Let $ESP(\alpha[\beta]\gamma,D)=(\alpha'[\beta]'\gamma',D\cup D')$ 
for a stable decomposition $S=\alpha\beta\gamma$.
There exist substrings 
$\underline{\alpha}$, 
$\underline{\alpha\beta}$, 
$\overline{\beta\gamma}$, $\overline{\gamma}$, each of whose change is 
at most $\logstar |S|+5$ such that
\begin{eqnarray*}
ESP([\alpha]\overline{\beta\gamma},D) &=& ([\alpha']y_1,D\cup D_1),\\
ESP(\underline{\alpha}[\beta]\overline{\gamma},D) &=& (x_2[\beta']y_2,D\cup D_2),\\
ESP(\underline{\alpha\beta}[\gamma],D) &=& (x_3[\gamma'],D\cup D_3), \mbox{ and}\\
D' &=& D_1\cup D_2 \cup D_3.
\end{eqnarray*}
\end{lemma}
\begin{proof}
Since $S=\alpha\beta\gamma$ is a stable decomposition of an ESP 
for $S$, the translated string $\alpha'$ and the dictionary $D_1$
for $D_1(\alpha')=\alpha$ are determined by only 
$\alpha$ and a prefix $\overline{\beta\gamma}$.
In case $p^+$ is the maximal prefix of $\beta\gamma$, 
we can set $\overline{\beta\gamma}=p^+$.
Otherwise, by Lemma~\ref{Cormode2}, 
we can set $\overline{\beta\gamma}$ to be a prefix
of length at most $\logstar |S|+5$.
For $\beta,\gamma$, we can set $\underline{\alpha}\overline{\gamma},
\underline{\alpha\beta}$ with the bounded change, respectively.
The above ESP defines $\alpha'(D_1)=\alpha$, $\beta'(D_2)=\beta$,
and $\gamma'(D_3)=\gamma$. 
By renaming all variables in the dictionaries,
there is a consistent $D'=D_1\cup D_2\cup D_3$ satisfying
$\alpha'(D')\beta'(D')\gamma'(D')=\alpha\beta\gamma$.~\qed
\end{proof}

\begin{lemma}\label{SLPtoESP}
Let $D$ be a dictionary of an SLP encoding a string $S\in\Sigma^*$.
A dictionary $D'$ of an ESP equivalent to $D$ is 
computable in $O(n\log^2 N+m)$ time,
where $n=|D|$, $m=|D'|$, and $N=|S|$.
\end{lemma}
\begin{proof}
We assume that $ESP^*(\derive(X_\ell),D)$ $(\ell\leq i,j)$ 
is already computed 
and let $D'$ be the current dictionary consistent with all $\derive(X_\ell)$.
For $X_k\to X_iX_j$ $(k>i,j)$, we estimate the time to update $D'$.
Let $\derive(X_i)=\alpha$ and $\derive(X_j)=\gamma$.

For the initial strings $\alpha,\gamma$,
we can obtain $\underline{\alpha}$ of length $\logstar N+6$
and $\overline{\gamma}$ of length $6$ in $O(\log N\logstar N)$ time.
By the result of $ESP(\underline{\alpha}\overline{\gamma},D')$,
we determine the position block $\beta$ which $\alpha[|\alpha|]$ and
$\gamma[1]$ belong to.
Then we can find a stable decomposition $S=\alpha\beta\gamma$
for the obtained $\beta$ and reformed $\alpha$ and $\gamma$,
where $\alpha$ (and $\gamma$) is represented by a path
from the root to a leaf in the derivation tree of $D_x$ (and $D_y$).
They are called a current tail and head, respectively.
Note that we can avoid decoding $\alpha$ and $\gamma$
for the parsing in Lemma~\ref{mergingESP}.
To simulate this, we use only the compressed representations $D_x,D_y$,
the current tail/head, and $\beta$.
Using the run-length representation, the change of $\beta$ is
bounded by $O(\logstar N)$ as follows.

By Lemma~\ref{mergingESP},
when $ESP(\alpha[\beta]\gamma,D)=(\alpha'[\beta]'\gamma',D\cup D')$ is 
computed by $ESP([\alpha]\overline{\beta\gamma},D) = 
([\alpha']y_1,D\cup D_1)$,
$ESP(\underline{\alpha}[\beta]\overline{\gamma},D) = 
(x_2[\beta']y_2,D\cup D_2)$, and 
$ESP(\underline{\alpha\beta}[\gamma],D) = 
(x_3[\gamma'],D\cup D_3)$,
the resulting string $\beta'$ is treated as the next $\beta$,
and the current tail and head are replaced by
$\alpha'[|\alpha'|]$ and $\gamma'[1]$ which represent
the next $\alpha$ and $\gamma$.

Let us consider the case
$ESP([\alpha]\overline{\beta\gamma},D) = ([\alpha']y_1,D\cup D_1)$.
If $\alpha\overline{\beta\gamma}$ contains a maximal repetition of $p$
as $\alpha[|\alpha|-N_1,|\alpha|]\cdot \overline{\beta\gamma}[1,N_2]=p^+$,
the next tail is the parent of $\alpha[|\alpha|-N_1-1]$,
which is determined in $O(\log^2 N)$ time since any repetition is 
replaced by the left aligned parsing and $N_1+N_2=O(N)$.
Otherwise, by Lemma~\ref{Cormode2},
we can determine the next tail by tracing 
a suffix of $\alpha$ of length at most $\logstar N+5$
in $O(\log N\logstar N)$ time.

Thus, $ESP(\alpha[\beta]\gamma,D)=(\alpha'[\beta]'\gamma',D\cup D')$ is 
simulated in $O(\log^2 N+m_k)$ for $X_k\to X_iX_j$,
where $m_k$ is the number of new variables produced in this ESP.
Therefore we conclude that 
the final dictionary $D'$ equivalent to $D$ is obtained 
in $O((\log^2 N+m_1)+\cdots+(\log^2 N+m_n))=O(n\log^2 N+m)$.~\qed
\end{proof}

\begin{theorem}\label{th1}\rm {(\bf SLP to Canonical SLP)}
Given an SLP $D$ of size $n$ for string $S$ of length $N$,
we can construct another SLP $D'$ of size $m$ 
in $O(n\log^2 N + m\log m\log^3 N)$
such that $D'$ is a final dictionary of $ESP^*(S,D')$
equivalent to $D$ and 
all variables in $D'$ are sorted by the lexical order
of their encoded strings.
\end{theorem}

\begin{theorem}\label{th77}\rm {(\bf Canonical SLP to LZ77)}
Given a canonical SLP $D$ of size $n$ for string $S$ of length $N$,
we can compute LZ77 factorization 
$f_1,\ldots,f_m$ of $S$ in $O(m\log^2n\log^3N+n\log^2n)$ time.
\end{theorem}
\begin{proof}
Using the technique in Lemma~\ref{sorting},
we can sort all variables $Z$ associated with $Z\to XY\in D$ 
by the following two keys:
the first key is the lexical order of $\derive(X)^R$ and 
the second is the lexical order of $\derive(Y)$,
where $S^R$ denotes the reverse string of $S\in\Sigma^*$.
Then $Z$ is mapped to a point $(i,j,pos)$ on a 3-dimensional space 
such that $i$ is an index of first key on $X$-axis,
$j$ is an index of second key on $Y$-axis, and
$pos$ is an index of leftmost occurrence of $\derive(Z)$ on $Z$-axis.
A data structure supporting range query for the point set
is constructed in $O(n\log^2 n)$ time/space and
achieving $O(\log^2 n)$ query time (See~\cite{Lueker78}).
Using this, we can compute $f_{\ell+1}$ from 
$f_1,\ldots,f_\ell$ and the remaining suffix $S'$
such that $S=f_1\cdots f_\ell\cdot S'$ as follows.

By Lemma~\ref{spire2011}, an evidence $Q=q_1\cdots q_k$ 
of $S'[1:j]$ satisfying $k=O(\log j)$ is found in $O(\log^2j)$ time.
Let $q_i$ be a symbol. Then we guess the division
$\alpha=q_1\cdots q_i$ and $\beta=q_{i+1}\cdots q_k$ to find
the range of $X$ in which $\alpha$ is embedded as its suffix,
the range of $Y$ in which $\beta$ is embedded as its prefix,
and the range of $Z$ whose leftmost position $pos$
satisfies $pos + j\leq |f_1\cdots f_\ell|$.
This query time is $O(\log^2n\log^2N)$.
Let $q_i=p_i^j$ for a symbol $p_i$.
Any maximal repetition is replaced by the left aligned parsing, 
and a resulting new repetition is recursively replaced by the same manner.
Thus, an embedding of $q_1\cdots q_k$ to $Z\to XY$ 
dividing $q_i=\alpha\beta$ such that
$q_1\cdots q_{i-1}\alpha$ is embedded to $X$ as suffix
and $\beta q_{i+1}\cdots q_k$ is embedded to $Y$ as prefix
is possible in $O(\log j)=O(\log N)$ divisions for $q_i$.
In this case, the query time is $O(\log^2n\log^3N)$.
Therefore, the total time to compute the required LZ77 factorization
is bounded by $O(m\log^2 n\log^3N+n\log^2 n)$.~\qed
\end{proof}

\begin{theorem}\label{th78}\rm {(\bf Canonical SLP to LZ78)}
Given a canonical SLP $D$ of size $n$ for string $S$ of length $N$,
we can compute LZ78 factorization 
$f_1,\ldots,f_m$ of $S$ in 
$O(m\log^3N+n)$ time.
\end{theorem}
\begin{proof}
Assume that the first $\ell$ factors $f_1,\ldots,f_\ell$ are obtained. 
By Lemma~\ref{spire2011},
we can find an evidence $Q_i$ of $f_i$ $(1\leq i\leq \ell)$,
and all evidences $Q_i$ $(1\leq i\leq \ell)$ are represented by a trie. 
Let $S'$ be the remaining suffix of $S$.
For each $j$, we can compute an evidence of $S'[1:j]$ in 
$O(\log^2j)=O(\log^2N)$ time.
Thus, we can find the greatest $j$ satisfying $f_i=S'[1:j]$
for some $1\leq i\leq \ell$ in $O(\log^3N)$ time using binary search.
Therefore, the total time to compute the required LZ78 factorization
is bounded by $O(m\log^3N+n)$.~\qed
\end{proof}

\begin{theorem}\label{RLE2ESP} \rm {(\bf Run Length Encoding to ESP)}
Given a text $S$ represented as a RL encoding $S=f_1\cdots f_n$ of length $n$,
we can compute an ESP $D$ representing $S$ in $O(n\logstar N)$ time.
\end{theorem}
\begin{proof}
We make a little change for replacing maximal repetition.
Consider maximal repetition $\alpha = a^k$ in $S$ is appeared.
If $k$ is even, then we replace $\alpha$ to $A^{k/2}$,
otherwise we replace to $A^{\lfloor k/2 \rfloor -1}B$ where $A \rightarrow aa$ and 
$B \rightarrow aaa$.
In exceptional case that the prefix and/or suffix of $\alpha$ is replaced with the left/right symbol adjacent to $\alpha$, 
we must consider for $\alpha^{\prime}$ removed such prefix/suffix from $\alpha$.
The computation time to replace such repetition is $O(1)$ since the number of repetitive symbols is represented as a integer.
Therefore, the time to convert is bounded by $O(n\logstar N)$.
\qed
\end{proof}

\begin{theorem}\label{ESP2Bisection}\rm {(\bf ESP to Bisection)}
Given an ESP $D$ of size $n$ representing $S$ of length $N$,
we can compute an SLP of size $m$ generated by Bisection in $O(m\log^2 N)$ time.
\end{theorem}
\begin{proof}
For each corresponding substring $S[i:j]$ under consideration.
We can obtain an evidence $Q$ corresponding to 
$S[i:i+2^k-1]$ in $O(\log^2 N)$ time.
By the Lemma~\ref{spire2011}, we can check if $Q$ is embedded as $T[i+2^k:j]$ in $O(\log^2 N)$ time.
If $Q$ can be embedded, we can allocate same variable for $S[i:i+2^k-1]$ and $S[i+2^k:j]$ since both substrings are equal, otherwise different variables are allocated.
Therefore, conversion can be done in $O(m\log^2 N)$ time.
\qed
\end{proof}

\full{
\section{Conclusions and Future Work}

In this paper we presented new efficient algorithms which, 
without explicit decompression, convert
to/from compressed strings represented in terms of run length encoding (RLE), 
LZ77 and LZ78 encodings, 
grammar based compressor RE-PAIR and BISECTION,
edit sensitive parsing (ESP),
straight line programs (SLPs), and admissible grammars.
All the proposed algorithms run in polynomial time in the input and output sizes,
while algorithms that first decompress the input compressed string
can take exponential time.
Examples of applications of our result are
dynamic compressed strings allowing for edit operations,
and post-selection of specific compression formats.

Future work is to extend our results to other text compression schemes,
such as Sequitur~\cite{SEQUITUR}, 
Longest-First Substitution~\cite{nakamura09:_linear_longest_first_},
and Greedy~\cite{apostolico00:_off_greedy_}.
}
\clearpage
\bibliographystyle{acm}
\bibliography{ref,compress,succinct}
\end{document}